\documentclass[ aps,pra, reprint,superscriptaddress, floatfix,  
nofootinbib,
nobibnotes,
amsmath,amssymb,
]{revtex4-2}

\usepackage{algorithm}
\usepackage{algpseudocode}
\usepackage[T1]{fontenc}
\usepackage{inputenc}
\usepackage{xcolor}
\usepackage{babel}
\usepackage{dsfont}
\usepackage{physics}
\usepackage{mathtools}
\usepackage{bbding}
\usepackage{pifont}
\usepackage{textcomp}
\usepackage{wasysym}
\usepackage{amsthm}
\usepackage[normalem]{ulem}
\usepackage{nicefrac}
\usepackage{wasysym}
\usepackage{dsfont}
\usepackage{amsmath}
\usepackage{amssymb}
\usepackage{graphicx}

\newcommand{\Ast}{\mathop{\scalebox{1.5}{\raisebox{-0.2ex}{$\ast$}}}}%

\renewcommand{\eqref}[1]{(\ref{#1})}

\usepackage{soul}
\definecolor{cream}{rgb}{1.0, 0.99, 0.82}

\definecolor{seafoam}{rgb}{	0.576, 0.914, 0.745}
\definecolor{olive}{rgb}{0.522, 0.666, 0.526}

\usepackage{hyperref}

\hypersetup{
     colorlinks   = true,
     linkcolor    = blue,
     citecolor    = blue,
     urlcolor     = blue,
     }

\usepackage{appendix}
\renewcommand{\selectlanguage}[1]{}
\makeatletter
\@ifundefined{textcolor}{}
{%
	\definecolor{BLACK}{gray}{0}
	\definecolor{WHITE}{gray}{1}
	\definecolor{RED}{rgb}{1,0,0}
	\definecolor{GREEN}{rgb}{0,1,0}
	\definecolor{BLUE}{rgb}{0,0,1}
	\definecolor{CYAN}{cmyk}{1,0,0,0}
	\definecolor{MAGENTA}{cmyk}{0,1,0,0}
	\definecolor{YELLOW}{cmyk}{0,0,1,0}
}
\theoremstyle{plain}

\theoremstyle{plain}

\ifx\proof\undefined
\newenvironment{proof}[1][\protect\proofname]{\par
	\normalfont\topsep6\p@\@plus6\p@\relax
	\Trivlist
	\itemindent\parindent
	\item[\hskip\labelsep
	\scshape
	#1]\ignorespaces
}{%
	\endtrivlist\@endpefalse
}
\providecommand{\proofname}{\textbf{Proof}}
\fi
\theoremstyle{plain}

\renewenvironment{proof}[1][\proofname]{\noindent {\bfseries #1.} }{\qed}

\providecommand{\lemmaname}{Lemma}
\providecommand{\definitionname}{Definition}
\providecommand{\propositionname}{Proposition}

\usepackage{babel}
\usepackage{txfonts}
\usepackage{colortbl}\definecolor{myurlcolor}{rgb}{0,0,0.7}

\renewcommand{\bra}[1]{\left\langle #1 \right|}
\renewcommand{\ket}[1]{\left| #1 \right\rangle}
\renewcommand{\braket}[2]{\left\langle #1 \middle| #2 \right\rangle}
\renewcommand{\ketbra}[2]{\left|#1\middle\rangle\!\middle\langle#2\right|}
\newcommand{\proj}[1]{\ketbra{#1}{#1}}

\newcommand{\id}{\mathds{1}}

\newcommand{\Ket}[1]{\left| #1 \middle \rangle\!\right \rangle}

\newcommand{\KetBra}[2]{\left|#1 \left \rangle\! \left \rangle\!\right \langle \! \right \langle #2\right |}
\newcommand{\Proj}[1]{\KetBra{#1}{#1}}

\newcommand{\haH}

\newtheorem{theorem}{Theorem}
\newtheorem*{thm*}{Theorem}

\newtheorem{conjecture}{Conjecture}

\newtheorem*{lem*}{Lemma}
\theoremstyle{definition}
\newtheorem{observation}{Observation}

\newtheorem{corollary}{Corollary}
\newtheorem*{cor*}{Corollary}

\newtheorem*{prop*}{Proposition}

\makeatother
\usepackage{subfiles}

\begin{document}
\preprint{APS/123-QED}
\title{Hamiltonian characterization of multi-time processes with classical memory}

\author{Kaumudibikash Goswami} 
\email{goswami.kaumudibikash@gmail.com}

\affiliation{QICI Quantum Information and Computation Initiative, School of Computing and Data Science, The University of Hong Kong, Pokfulam Road, Hong Kong.}

\author{Abhinash Kumar Roy}
\thanks{K.G. and A.K.R contributed equally.}

\affiliation{Department of Physical and Mathematical Sciences, Macquarie University, Sydney NSW, Australia.}

\author{Varun Srivastava}

\affiliation{Department of Physical and Mathematical Sciences, Macquarie University, Sydney NSW, Australia.}

\author{Barr Perez} 

\affiliation{Department of Physical and Mathematical Sciences, Macquarie University, Sydney NSW, Australia.}%

\author{Christina Giarmatzi}

\affiliation{Department of Physical and Mathematical Sciences, Macquarie University, Sydney NSW, Australia.}

\affiliation{School of Computer Science, University of Technology Sydney, Sydney NSW, Australia}

\author{Alexei Gilchrist}
\affiliation{Department of Physical and Mathematical Sciences, Macquarie University, Sydney NSW, Australia.}

\author{Fabio Costa}
\email{fabio.costa@su.se}
\affiliation{Stockholm University and KTH Royal Institute of Technology, Stockholm, Sweden.}

\begin{abstract}
 A central problem in open quantum systems is the characterization of non-Markovian processes, where an environment retains the memory of its interaction with the system. A key distinction is whether or not this memory can be simulated classically, as this can lead to efficient modelling and noise mitigation. Powerful tools have been developed recently within the process matrix formalism, a framework that conveniently characterizes all multi-time correlations through a sequence of measurements. This leads to a detailed classification of classical and quantum-memory processes and provides operational procedures to distinguish between them. However, these results leave open the question of what type of system-environment interactions lead to classical memory.
More generally, process-matrix methods lack a direct connection to joint system-environment evolution, a cornerstone of open-system modelling. 
In this work, we characterize Hamiltonian and circuit-based models of system-environment interactions leading to classical memory. We show that general time-dependent Hamiltonians with product eigenstates, and where the environment's eigenstates form a time-independent, orthonormal basis, always produce a particular type of classical memory: probabilistic mixtures of unitary processes. Equivalently, these Hamiltonians are characterized as commuting with a complete set of observables on the environment. Additionally, we show that the most general type of classical memory processes can be generated by a quantum circuit in which the system and environment interact through a specific class of controlled unitaries.
Our results establish the first strong link between process-matrix methods and traditional Hamiltonian-based approaches to open quantum systems. 
\end{abstract}

\maketitle

 \noindent 
 \section{Introduction}
Characterizing non-Markovian noise in the quantum regime has been a notoriously difficult problem. While this is well-developed for classical stochastic processes, quantum non-Markovianity is a more elusive concept with diverse definitions found in the literature \cite{LI20181,carmichael_1993,Breuer_book_OQS,LI20181}, with approaches like modelling the system-environment interactions~\cite{Lax_1963,Lax_1966,gardiner_2004,Breuer_2009,non_markovianity_entropic,Buscemi2014, Bylicka_2017, Buscemi2025} and dynamical maps ~\cite{kossakowski_1972, gorini_1976,lindblad_1976,Rivas2014,Rivas_2010, Rivas2014}.

Recently, an operational approach to quantum non-Markovianity was developed~\cite{Pollock_pra,pollock_operational_markov, Shrapnel2018_supervised_learning, Luchnikov2019} and deployed in experiments~\cite{White2020,Goswami_non_Markov,Guo2021, modi_ibm,Liang2021, White2022, giarmatzi2023multitime, White2025whatcanunitary}. It provides a complete characterization of a multi-time quantum process: a system evolving in time and probed through a sequence of measurements, where an environment can mediate correlations between measurements at different times. 
This characterization is represented by the \emph{process matrix}~\cite{oreshkov12, oreshkov15}. Markovianity of the multi-time process is reflected by conditions on the process matrix. 
Non-Markovian processes are classified into two main classes depending on whether the memory required to model the environment is classical or quantum, and the former can be further subdivided into more specialized classes~\cite{Giarmatzi2021witnessingquantum,Nery_2021,taranto2023characterising}.

However, the process matrix approach drastically differs from the traditional approach, where the open-system dynamics are modelled by a continuous-time evolution governed by a system-environment Hamiltonian. Rather, it describes the dynamics through a fixed set of gates in a quantum circuit. 
The system can be probed through operations between these gates, whose outcomes reconstruct the process matrix. Although the process matrix formalism confers many operational advantages~\cite{Modi_2012, pollock_operational_markov,Wu2023}, the link between the two approaches has been little explored. A particularly relevant question is how to characterize types of memory in multi-time quantum processes in terms of system-environment Hamiltonians.

Here, we focus on processes that can be simulated with a classical memory, which we call classical memory processes. 
We characterize these processes using both Hamiltonian and quantum-circuit-based approaches. 
In the Hamiltonian-based approach, we find a class of Hamiltonians that leads to a particular class of classical-memory processes given by a probabilistic mixture of unitary (hence Markovian) processes. The result holds irrespective of the choice of probing times, initial system-environment correlation, and the number of time steps. Interestingly, the Hamiltonian class includes entangling interactions such as the Ising interaction. Importantly, establishing this connection between the Hamiltonian and the memory in the process clarifies when an interaction enables device-independent memory detection \cite{Roy_PRA_2024}, as well as when it remains invisible to standard randomized benchmarking protocols \cite{srivastava2025blindspotsrandomizedbenchmarkingtemporal}.

In the quantum-circuit-based approach, we show that any such probabilistic mixture of unitary processes has a dilation in terms of controlled unitaries, with the environment being the control in the same basis throughout the process. Further, in Appendix~\ref{App:unitary_form_classical_memory}, we show that the most general classical-memory processes~\cite{Giarmatzi2021witnessingquantum, taranto2023characterising}  have a similar dilation with a more general class of controlled-unitaries, where additional environmental degrees of freedom can produce stochastic conditional memory.

\noindent 
\section{Multi-time process framework}
We utilize the framework of multi-time process~\cite{chiribella09b,Pollock_pra}, represented by the process matrix~\cite{oreshkov12}, to encode the most general system-environment interactions. Just as a quantum map captures the dynamics of an open quantum system by providing two-time correlations between the initial and final states without explicitly specifying the underlying system–environment dynamics, the process matrix can be understood as a higher-order quantum map in a multi-time setting. In this framework, interventions at each time step probe the system and are represented by quantum operations. Whereas an ordinary quantum map takes a quantum state as its input, this higher-order map instead takes the sequence of quantum operations applied at each time step and yields a joint probability distribution that encodes the multi-time correlations between the probes, thereby offering an operational description of quantum stochastic processes \cite{pollock_operational_markov}. Moreover, unlike quantum maps, this framework allows arbitrary correlation between the system and environment states. Furthermore, techniques analogous to quantum process tomography can be employed to experimentally reconstruct the process matrix describing a multi-time process \cite{giarmatzi2023multitime}, see Appendix~\ref{Sec:Preliminaries} for details.

As shown in Fig.~\ref{fig:gen_non-Markov}, an $N$-time process is associated with $N$ ``sites'', where each site $n$, representing the system at time $t_n$, can be probed with a quantum operation represented by a completely positive (CP) trace non-increasing map, $\mathcal{M}_{a_n|x_n}:A_I^n {\to} A_O^n$, with $x_n$ and $a_n$ being the classical setting and outcome respectively, and $A_I^n$ and $A_O^n$ being the input and output systems respectively. Here, $A_{I(O)}^n{\equiv}\mathcal{L}(\mathcal{H}^{A^n_{I(O)}})$ is the space of linear operators in the Hilbert space $\mathcal{H}^{A^n_{I(O)}}$. Note that we assume the time scale to perform the probing operations to be much smaller than the system-environment interaction time. The map $\mathcal{M}_{a_n|x_n}$ is equivalent to its Choi-Jamio{\l}kowski (CJ) representation~\cite{choi_completely_1975,jamio72,choi75b} $M_{a_n|x_n}:=\sum_{i,j}\ketbra{i}{j}{\otimes}\mathcal{M}_{a_n|x_n}(\ketbra{i}{j})$ defined on $A_I^n\otimes A_O^n$, where $\{\ket{i}\}_i$ is an orthonormal basis at $A_I^n$. The CP condition of the map $\mathcal{M}_{a_n|x_n}$ is equivalent to the positive semi-definiteness of its CJ representation, $M_{a_n|x_n}\ge 0$, and the trace non-increasing condition is guaranteed by the inequality $\Tr_{A_O^n}M_{a_n|x_n}\le \mathds{1}^{A_I^n}$. Moreover, summing over all the classical outcomes of $\mathcal{M}_{a_n|x_n}$ results in a completely positive trace-preserving (CPTP) map, which is reflected in its CJ representation as $\sum_{a_n}M_{a_n|x_n}=\mathds{1}^{A_I^n}$.

The probability of the classical outcomes $\vec{a}{:=}\{a_1,a_2,{\cdots}, a_n\}$ conditioned on the settings $\vec{x}{:=}\{x_1,x_2,{\cdots}, x_n\}$ is given by the generalised Born rule~\cite{oreshkov12,Shrapnel_2018}:  

\begin{align}
p(\vec{a}|\vec{x})=\Tr\left[\left(M_{a_1|x_1}\otimes M_{a_2|x_2}\otimes \cdots \otimes M_{a_n|x_n}  \right)^T W\right].
\end{align}
Here, $W{\in} A_I^1{\otimes} A_O^1{\otimes} A_I^2{\otimes} A_O^2{\otimes} {\cdots} {\otimes} A_I^N{\otimes} A_O^N $ is a positive semidefinite operator, called the process matrix~\cite{oreshkov12}, which fully characterises the background system-environment interaction.

In a Markovian process, the environment resets after each time step, resulting in memoryless dynamics: the system's evolution is independent of its past states. We show such a process in  Fig.~\ref{fig:gen_non-Markov} (b), which is composed of an initial system-environment state $\rho\in A_I^1E^1$ and quantum channels $\mathcal{T}^n{:}A^n_O{\to} A_I^{n{+}1}$ for $n\in \{1,\cdots, N{-}1\}$, defined as $\Tr_{E^{n+1}_I} \mathcal{T}^n(\rho):=U^n[\rho{\otimes}\tau](U^{n})^\dagger$. Here, $U^n$ denotes the system-environment unitary at the $n$-th time step, and $\tau \in E^n$ is the local environment state. The corresponding multi-time process $W_M$ is 
\begin{align}
    W_M:=\rho_{A_I}^1 \otimes T^1 \otimes \cdots \otimes T^{N-1},
\end{align}
where $\rho^{A_I^1}:=\Tr_{E^1}\rho$, and $T^n$ is the CJ representation of the quantum channel $\mathcal{T}^n$. A particular case of Markovian process, relevant to us, is the \emph{unitary Markovian process matrix} $W_U$ where each channel $\mathcal{T}^n$ is unitary, $\mathcal{T}^n(\rho) = V^n\rho (V^n)^{\dagger}$, $V^n(V^n)^{\dagger} = (V^n)^{\dagger}V^n = \id$ with its CJ representation $\Proj{{V}^n}\in A_O^n\otimes A_I^n$, where $\Ket{{V}^n}:=\sum_i\ket{i}\otimes{V}^n\ket{i}$:
\begin{align}
W_U:=\rho^{A_I^1}\otimes \Proj{{V}^1}\otimes \cdots \otimes \Proj{{V}^{N-1}}.\label{Eq:Markov_unitary_main_text}
\end{align}

Now, we introduce the classical common cause (CCC) process. The nomenclature is inspired by quantum causal modelling \cite{Ried_2015, costa2016, Allen_2017, barrett2020quantum} and captures the idea that the memory effects can be attributed to a single, latent common cause, simultaneously affecting all observed events~\cite{direct_cause_note}. A CCC process $W_{\mathrm{CCC}}$ can be represented as a convex sum of Markovian processes: 
\begin{align}
W_{\mathrm{CCC}}&=\sum_{\nu}p(\nu) \rho_{\nu}\otimes T^{1}_\nu \otimes \cdots \otimes T^{N-1}_{\nu}=\sum_\nu p(\nu){W}_\nu.\label{Eq:CCC_main_text}
\end{align}
Here, each ${W}_\nu:=\rho_\nu\bigotimes_{n=1}^{N-1}T^{n}_\nu$ is a Markovian process. A particular case of CCC process relevant to us is the \emph{mixed unitary process} where each ${W}_\nu$ in Eq.~\eqref{Eq:CCC_main_text} is a unitary Markovian process as in Eq.~\eqref{Eq:Markov_unitary_main_text}.

\begin{figure}
    \centering
    \includegraphics[width=\columnwidth]{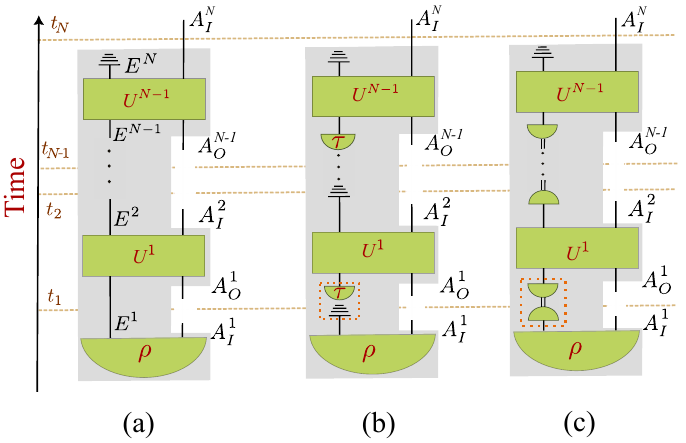}
    \caption{Non-Markovian processes with different memories. (a) The general non-Markovian process with initial joint system-environment state $\rho^{A_I^1E^1}$ and subsequent system-environment unitary interaction $U^n:A_O^nE^n{\to} A_I^{n{+}1}E^{n+1}$ for $n{\in}\{1,2,{\cdots},{N-1}\}$.(b) A Markovian process, A Markovian process, where the environment is reset to a fresh state $\tau$ at each timestep via trace-and-replace channels (orange dashed box), preventing memory retention.  (c) A classical-memory process, where an entanglement-breaking channel acts on the environment at each step (orange dashed box), preserving only classical correlations across time. The ``gaps'' in the circuits represent the possibility of performing arbitrary operations on the system at the given times, under the assumption that their duration is negligible compared to the interaction time. }
    \label{fig:gen_non-Markov}
\end{figure}

\noindent  

\section{Results}
\subsection{Hamiltonians leading to CCC processes} 

In this work, we are interested in relating multi-time processes to an underlying system-environment dynamics, governed by a time-dependent Hamiltonian $H(t)$ acting on a product space $\mathcal{H}^S\otimes \mathcal{H}^E$. Accordingly, the joint system-environment unitary operator $U^n$ from time $t_n$ to $t_{n+1}$ is given by a time-ordered exponential~\cite{sakuraiModernQuantumMechanics2021}
\begin{align}\label{Eq:nunitary}
  U^n = T\exp\left[-i \int_{t_n}^{t_{n+1}}dt H(t)\right].
\end{align}
The global, unitary, $N$-time process is given by Eq.~\eqref{Eq:Markov_unitary_main_text}, where $\Proj{U^n}\in A^n_O \otimes E^n_O \otimes A^{n+1}_I \otimes E^{n+1}_I$, and all input-output spaces are isomorphic to the system-environment space: $\mathcal{H}^{A^{n}_{I(O)}}\cong \mathcal{H}^S$, $\mathcal{H}^{E^{n}_{I(O)}}\cong \mathcal{H}^E$. To find the reduced process describing the system alone, we need to insert identity operations on the environment at each time and discard the final environment space $E^n_I$. Following \cite{chiribella09b}, this can be done by first identifying each input and output environment space, which means redefining the Choi operators of the unitaries as $\Proj{U^n}\in A^n_O \otimes E^n \otimes A^{n+1}_I \otimes E^{n+1}$, with $\mathcal{H}^{E^{n}}\cong \mathcal{H}^E$, and then calculating the system process matrix $W^{A^{1}_IA^{1}_O{\dots}A^{N}_I}$ as
\begin{align}
&W^{A^{1}_IA^{1}_O{\dots}A^{N}_I}=\rho^{A_{I}^{1}E^{1}} \bigg(\Ast _{n=1}^{N-1}\Proj{U^{n}}^{A^{n}_OE^{n}A^{n+1}_IE^{n+1}}\bigg)
       {*}\mathds{1}^{E^{N}}, \label{Eq:process_matrix}
\end{align}
where we have used the \emph{link product}~\cite{chiribella09b}, `*' which allows us to compose arbitrary quantum operations, e.g., state preparation, transformation via channels or measurement, in terms of their CJ representations. In particular, suppose $P$ and $Q$ are CJ representations of two quantum operations with $\mathcal{P}$ and $\mathcal{Q}$ being their respective Hilbert spaces, then the CJ representation of their composition is given by

\begin{align}
    P * Q := \Tr_{\mathcal{P}\cap \mathcal{Q}}[(\mathds{1}^{\mathcal{Q}\setminus \mathcal{P}}{\otimes}P^{T_{\mathcal{P}\cap \mathcal{Q}}})(Q{\otimes}\mathds{1}^{\mathcal{P}\setminus \mathcal{Q}})]. 
\end{align}
In other words, the formula composes the operations by connecting them on their shared subsystems, $\mathcal{P}\cap \mathcal{Q}$, which involves a partial transpose $T^{\mathcal{P}\cap \mathcal{Q}}$ over these shared spaces, followed by a partial trace, $\Tr_{\mathcal{P}\cap \mathcal{Q}}$, over them. For us, the shared spaces are the environment spaces $E^n$. We can now state our main result:

\begin{theorem} \label{Thm:sufficient_Hamiltonians}
    If a time-dependent Hamiltonian $H(t)$ admits a structure 
    \begin{equation} \label{CCCHamiltonian}
    H(t){=}\sum_{j}S_{j}(t)\otimes \mathcal{E}_{j}(t),  
    \end{equation}
    where the operators $S_{j}(t)$ act on the system and $\mathcal{E}_{j}(t)$ on the environment, and $\Big[\mathcal{E}_{j}(t),\mathcal{E}_k(\bar{t})\Big]{=}0 $ for all $j, k , t$ and $\bar{t}$, then the process matrix $W$ in Eq.~\eqref{Eq:process_matrix} is a mixed unitary process. Specifically, it has the form
\begin{align}\label{maintheoremmixedunitary}
W^{A^{1}_IA^{1}_O{\dots}A^{N}_I} = \sum_\nu p(\nu)  \widetilde{W}_\nu ^{A^{1}_IA^{1}_O{\dots}A^{N}_I},
\end{align}
where each $\widetilde{W}_\nu$ is a unitary Markovian process as in Eq.~\eqref{Eq:Markov_unitary_main_text}, and $\{p(\nu)\}_{\nu}$ is a probability distribution.
\end{theorem}
See Appendix~\ref{App:proof_theorem_suff_hamil} for the proof. Interestingly, the Hamiltonian in Theorem~\ref{Thm:sufficient_Hamiltonians} has product eigenstates between the system and environment for all times. This is because the operators $\mathcal{E}_j(t)$ all commute, so they admit an eigendecomposition $\mathcal{E}_j(t){=}\sum_\nu\lambda_j^\nu(t)\proj{\nu}$, where $\{\lambda _j^\nu\}_\nu$ are the eigenvalues and the eigenstates $\{\ket{\nu}\}_\nu$ are independent of both time and the index $j$. Thus we have 
\begin{align}
 H(t)=\sum_{\nu}\Tilde{S}^\nu(t)\otimes \proj{\nu},\ \mathrm{with}\    \Tilde{S}^\nu(t):=\sum_j\lambda_j^\nu(t)S_j^\nu(t)  \label{Eq:Heigenstates1}
\end{align}
Finally, considering the eigendecomposition  $\Tilde{S}^\nu(t){=}\sum_{\mu} \omega^{\mu\nu}(t)\proj{\psi^{\mu\nu}(t)}$, we observe 
\begin{equation}
    \label{Heigenstates}
H(t) = \sum_{\mu,\nu} \omega^{\mu\nu}(t)\proj{\psi^{\mu\nu}(t)}\otimes \proj{\nu}.
\end{equation}

The fact that Hamiltonians leading to CCC have product eigenstates that are time-independent on the environment eigenstates immediately leads to an equivalent, physically relevant, characterization: 

\begin{corollary}
 A Hamiltonian $H(t)$ admits a decomposition as in Eq.~\eqref{CCCHamiltonian} if and only if there is a complete set of conserved observables in the environment. That is to say, the Hamiltonian commutes with a set $\{F_n\}_n $ of environment operators with a non-degenerate common eigenbasis. 
\end{corollary}

\begin{proof}
In the forward implication, a Hamiltonian in Eq.~\eqref{CCCHamiltonian} has eigendecomposition in the form in Eq.~\eqref{Heigenstates}, so any environment observable of the form $F_n{=}\sum_{\nu} f_n(\nu) \ketbra{\nu}{\nu}$, $f_n(\nu){\neq} f_n(\nu')$ for $\nu{ \neq} \nu'$, is conserved (commutes with the Hamiltonian) and complete (it has a non-degenerate eigenbasis). Conversely, if the Hamiltonian commutes with a complete set of environment observables, it must have eigenstates of the form in Eq.~\eqref{Heigenstates}. 
\end{proof}

Physically, the operators $\{F_n\}$ represent conserved quantities of the environment, such as energy, spin components, particle number, etc. Their conservation, $[H(t), F_n] = 0$, implies that their eigenstates are invariant under evolution, i.e., if the environment is initialized in an eigenstate $\ket{\nu}$ of these observables, it remains in that state throughout the evolution. This identifies the eigenbasis $\{\ket{\nu}\}$ as a set of invariant, non-evolving states that classically label the distinct unitaries applied to the system.

Theorem~\ref{Thm:sufficient_Hamiltonians} provides a sufficient condition for a \emph{time-dependent} Hamiltonian to generate a mixed unitary process exhibiting classical memory for \emph{any} choice of probing times $t_1,{\cdots},t_N$. As shown in Eq.~\eqref{Heigenstates}, a key feature of such a Hamiltonian is that it has product eigenstates. This naturally leads us to investigate the converse: is this product structure also a \emph{necessary} condition for a Hamiltonian to guarantee vanishing quantum memory for all possible probing choices? It turns out that the answer depends on whether the Hamiltonian is time-dependent or time-independent.

We have first considered time-dependent Hamiltonians and found a counterexample violating this necessity. Specifically, we have constructed a ``pulsed'' Hamiltonian with entangled eigenstates that results in a Markovian process for all probing times.

\begin{observation}
Consider the time-dependent ``pulsed'' Hamiltonian $H(t) = \bar{H} \sum_i \delta(t-t^*_i)$ with $\delta(\cdot)$ being the Dirac delta function and where $t^*_1,t^*_2,\dots$ are fixed times at which the ``pulses'' occur. Using Dyson series expansion, one can check that the corresponding unitary between times $t_n$ and $t_{n+1}$ is $U^n = e^{-i \bar{H}}\cdots e^{-i \bar{H}} = e^{-i m \bar{H}}$, where $m$ is the number of pulse times $t^*_j$ between $t_n$ and $t_{n+1}$. Now, considering $\bar{H} {=} 2\pi H_{\mathrm{SWAP}}$, and noting $e^{-i \theta H_{\mathrm{SWAP}}} = \cos(\theta) \mathds{1} -i \sin(\theta) H_{\mathrm{SWAP}}$, we have $U^n=\mathds{1}$ for all $n$, i.e., the corresponding process is a Markovian process for all choices of time even though the underlying time-dependent Hamiltonian has entangled eigenstates. 
However, we also note that this Hamiltonian is operationally indistinguishable from a trivial Hamiltonian: $H(t){=}0$ for all $t$. This leads us to speculate that our necessary condition can be rephrased to include time-dependent Hamiltonians by excluding degenerate cases, as in this example. 
 \end{observation}

We then turn to the more constrained case of \emph{time-independent} Hamiltonians. Here, we are unable to find any such counterexamples. On the contrary, we have found examples where time-independent Hamiltonians with entangled eigenstates kill off quantum memory only for fine-tuned, discrete sets of probing times, while generically exhibiting quantum memory.

\begin{observation} \label{Obs:entangled_eigenstate}
Consider the time-independent Hamiltonian represented by a SWAP operator. This Hamiltonian has entangled eigenstates, leading to a unitary process for a particular sequence of probing times. However, the same Hamiltonian results in a multi-time process with quantum memory for other choices of probing times.  Specifically, the 2-qubit SWAP Hamiltonian has the eigendecomposition $H_{\mathrm{SWAP}}=\sum_{i=0}^3 \alpha_i \proj{\Phi_i}$, where $\ket{\Phi_i}$ are four Bell-states: 
\begin{align}
&\ket{\Phi_0}=\frac{1}{\sqrt{2}}(\ket{00}+\ket{11}), \\
&\ket{\Phi_1}=\frac{1}{\sqrt{2}}(\ket{01}+\ket{10}), \\
&\ket{\Phi_2}=\frac{1}{\sqrt{2}}(\ket{01}-\ket{10}), \ \mathrm{ and} \\
&\ket{\Phi_3}=\frac{1}{\sqrt{2}}(\ket{00}-\ket{11}),
\end{align}
with $\alpha_{i=2}=-1$ and $\alpha_{i\ne 2}=1$. Now, for a choice of times $\{t_n=2\pi n\}_n$, we have $U_n=\exp(-it_n H_{\mathrm{SWAP}})=\mathds{1}$, which is a 2-qubit identity map. Hence, the process defined for these probing times is a Markovian process (hence a subclass of a CCC process), although the Hamiltonian has maximally entangled eigenstates. However, a different choice of times reveals the non-trivial memory in the process. For instance, for a choice of times $\{t_n=(4n+1)\pi/2\}_n$, we have $U_n=\exp(-it_n H_{\mathrm{SWAP}})=-i{H_{\mathrm{SWAP}}}$. As discussed in Ref.~\cite{Giarmatzi2021witnessingquantum}, a process defined by a swap between system and environment is a quantum memory process.
\end{observation} 
In a similar spirit to Observation~\ref{Obs:entangled_eigenstate}, in Appendix~\ref{App:Hamiltonian_entangled}, we show an example of the Heisenberg model Hamiltonian having entangled eigenstates that can lead to a mixed unitary process for a specific choice of times but exhibits a quantum memory process in general.
The lack of a counterexample for the time-independent case leads us to propose the following conjecture:

\begin{conjecture}
   Every time-independent Hamiltonian that generates a CCC process for all choices of probing times has the form
   \begin{equation}
       H = \sum_{\mu \nu} \omega^{\mu \nu} \proj{\psi^{\mu \nu}}\otimes \proj{\nu}.
   \end{equation}
\end{conjecture}   

We note in passing that the constraint in Theorem \ref{Thm:sufficient_Hamiltonians} essentially requires the eigenstates of the environment to be time-independent. In Appendix \ref{App:varying_env-basis}, we present a scenario where the environment eigenstates are slowly time varying and show that even a slight variation, in general results in quantum memory.

\noindent 
\subsection{Quantum circuit model for CCC process}

An interesting question is how to simulate CCC processes on gate-based quantum computers, where the parties act at fixed times and the system-environment interaction is modelled by a fixed sequence of unitaries. The answer lies in the unitary associated with the Hamiltonian in Theorem~\ref{Thm:sufficient_Hamiltonians}, which for any timesteps $t_n$ to $t_{n+1}$ is given by the time-ordered exponential $U(t_n,  t_{n+1})=T\exp[-i\int_{t_n}^{t_{n+1}}dt H(t)]$. Expanding $H(t)$ as in Eq.~\eqref{Eq:Heigenstates1}, we have
\begin{align}
 U(t_n,  t_{n+1})&=\sum_{\nu}T\exp[-i\int_{t_n}^{t_{n+1}}dt\Tilde{S}^\nu(t)] \otimes \proj{\nu}  \nonumber \\
 &=\sum_{\nu} \widetilde{U}_{\nu}(t_n,  t_{n+1})\otimes \proj{\nu},
\end{align}
where $\widetilde{U}_{\nu}(t_n,  t_{n+1}){=}T\exp[-i\int_{t_n}^{t_{n+1}}dt\Tilde{S}^\nu(t)]$. In other words, the unitary takes the form of a controlled unitary, with the environment being the control system in the same orthonormal basis throughout all the time steps. Thus, according to Theorem~\ref{Thm:sufficient_Hamiltonians}, such unitaries result in a CCC process (specifically, a mixed unitary process).

\begin{figure}
    \centering
    \includegraphics[width=\columnwidth]{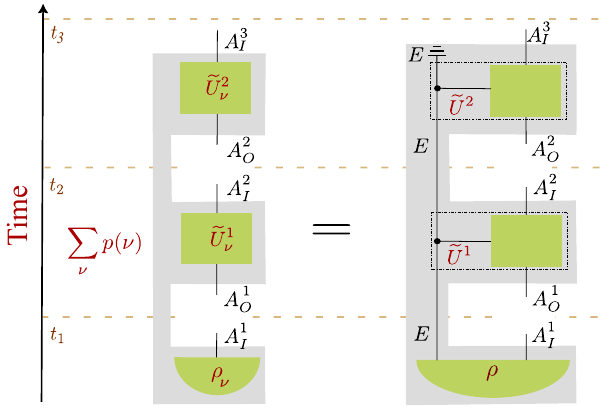}
    \caption{ Depiction of Theorem~\ref{Thm:direct_cause_extension} for a three-time mixed unitary process $W{=}\sum_{\nu} p(\nu) W_{\nu}$ with $W_{\nu} {=} \rho_{\nu}^{A_I^1}{\otimes}\Proj{\widetilde{U}_{\nu}^1}^{A_O^1A_I^1}{\otimes}\Proj{\widetilde{U}_{\nu}^2}^{A_O^2A_I^3}$ being a Markovian process. Here $\rho_{\nu}$ is an initial state at $A_I^1$ the system, and $\Proj{\widetilde{U}_{\nu}^k}$ is the CJ representation of the unitary $\widetilde{U}_{\nu}^k$, with $k {\in} \{1,2\}$. The right-hand side of the picture shows that any such mixed unitary process can be dilated in terms of controlled unitaries $\widetilde{U}^k=\sum_{\nu}\widetilde{U}_{\nu}^k\otimes \proj{\nu}^E$, $k {\in} \{1,2\}$, and an initial quantum-classical system-environment state $\rho^{A^1_IE}=\sum_{\nu}p(\nu)\rho_{\nu}^{A_I^1}{\otimes}\proj{\nu}$. See Appendix~\ref{App:proof_direct_cause_ext} for the proof.}
    \label{fig:CDC_necessary}
\end{figure}

In the following theorem, we show that the converse statement also holds, and thereby formalize the connection between controlled unitaries and the mixed unitary processes.  

\begin{theorem} \label{Thm:direct_cause_extension}
        All mixed unitary processes have a dilation in terms of controlled unitaries, with the environment being the control with the same orthonormal basis in all time steps.
\end{theorem}

See Appendix~\ref{App:proof_direct_cause_ext} for the proof. Fig.~\ref{fig:CDC_necessary} depicts Theorem~\ref{Thm:direct_cause_extension} for a mixed unitary process with $N{=}3$. Beyond mixed unitary processes, in Appendix~\ref{App:unitary_form_classical_memory}, we also characterize the quantum-circuit models, involving nested-controlled unitaries, that give rise to the most general classical-memory processes with fixed probing times. The intuition behind this form is that a multi-time process is a classical-memory process if it can be modelled by a sequence of system and environment interactions that are conditional quantum instruments propagating classically-encoded information across sites. We explicitly construct dilations of these instruments which are conditional
unitary with an extra environment, and a measurement on that
environment as shown in Fig.~\ref{fig:stochastic_to_control}. Note that, while we have shown that all classical memory processes admit at least one dilation based on controlled or nested-controlled unitaries, such constructions are not unique, and alternative constructions resulting from a different dilation are possible. We now present some examples of Hamiltonians that satisfy the conditions of Theorem \ref{Thm:sufficient_Hamiltonians}, resulting in classical-memory processes. 

\noindent 
\subsection{Examples}

\textit{Example 1:} A simple two-qubit model, with one of the qubits as system and another one as environment, with interaction Hamiltonian given by $H^{SE} = w\sigma_{i}\otimes \sigma_{j}$ where $\sigma_{i}$ are Pauli operators and $w$ is the interaction strength, has been well studied in the context of non-Markovianity \cite{Jordan_PRA_2004,UshaDevi_PRA_2011,Wudarski_2016_EPL}. Since the Hamiltonian is of the product form, Theorem \ref{Thm:sufficient_Hamiltonians} implies the resulting process will be a convex sum of Markovian processes, i.e., a process with classical memory. In particular, if the initial system-environment state is $\rho^{A^{1}_{I}E_1}$, then the resulting process will be of the form,
\begin{align}
    W = \sum_{\nu} p(\nu) \rho_{\nu}^{A_{I}^{1}}\otimes T_{\nu}^{A_{O}^{1}A_{I}^{2}}\otimes T_{\nu}^{A_{O}^{2}A_{I}^{3}}\otimes \mathds{1}^{A_{O}^{3}},
\end{align}
where
\begin{align}
    p ({\nu}) \rho_{\nu}^{A_{I}^{1}} = (\id^{A^{1}_{I}}\otimes\langle \lambda_{\nu}^{j}|)\rho^{A^{1}_{I}E_1}(\id^{A^{1}_{I}}\otimes|\lambda_{\nu}^{j}\rangle)
\end{align}
and the unitary channels $T_{\nu}$ are 
\begin{align}
  T_{\nu} =\! \sum_{m,n} c_{\nu m} c_{\nu n}^{*}|\lambda_{m}^{i}\lambda_{m}^{i}\rangle\langle\lambda_{n}^{i}\lambda_{n}^{i}| = \Proj{e^{-i w \lambda_{\nu}^{j} \sigma_i\Delta t}},
\end{align}
where $c_{\nu m } = \exp(-iw\lambda_{\nu}^{j}\lambda_{m}^{i}\Delta t)$, $\nu, m, n$ can be $+,-$ corresponding to positive and negative eigenvalues of corresponding Pauli operators, and $\Delta t$ is the time interval between sites (which we take to be all equal for simplicity).  A special example of the above Hamiltonian is the ZZ-type interaction Hamiltonian, which leads to cross-talk in transmon-based quantum devices. Note that the result also applies to an arbitrarily large environment: as long as all interactions are of $ZZ$ type, an arbitrary number of environment qubits leads to classical memory. In contrast, some recent superconductor-based experiments, such as Refs.~\cite{giarmatzi2023multitime, White2025whatcanunitary}, demonstrate nonclassical memory, which indicates in those cases a ZZ-type Hamiltonian cannot explain the non-Markovian noise.

\medskip

\noindent \textit{Example 2:} A widely used global operation in quantum information processing is the CNOT gate, which is generated by the Hamiltonian $H^{SE} {=}  \proj{-}^S {\otimes} \proj{1}^{E}$, with $U_{\mathrm{CNOT}}=\exp[-i\pi H]$. This Hamiltonian is in the product form; therefore, it has fixed (time-independent) product eigenstates. Due to Theorem~\ref{Thm:sufficient_Hamiltonians}, the CNOT operation can not generate a quantum non-Markovian process. 

\medskip

\noindent \textit{Example 3:} A case where the environmental degree of freedom is not finite is the following model: a photon's polarisation degree of freedom is the system qubit, and frequency modes are the environment. The interaction between these degrees of freedom is given by the experimentally implementable birefringent effect,
\begin{equation}
    H^{SE} = \int_0^{\infty}  \omega  (\eta_{0}|0\rangle\langle 0| + \eta_{1}|1\rangle\langle 1|) \otimes |\omega\rangle\langle\omega| \mathrm{d}\omega.
\end{equation}
This system has been studied to experimentally realize and control non-Markovian dynamics, particularly Markovian to non-Markovian transition \cite{Liu_2011_Nat_Phys}. The structure of Hamiltonian satisfies the criteria of Theorem \ref{Thm:sufficient_Hamiltonians}.
Therefore, the resulting process is a mixed-unitary process.

\noindent \textit{Example 4:} In the diamond NV centre, the full Hamiltonian between the electronic spin triplet and the nitrogen nuclear spin (of $^{14}N$ isotope) is given by the Hamiltonian of the form 
\begin{align}
    \label{NV centre}
    H^{SE}&=\gamma_{1}S_{z}^{2} \otimes \id + \gamma_{2} S_z \otimes \id + \gamma_3 \id \otimes I_{z}^{2} + \gamma_{4} \id \otimes I_{z}+\\
    &~~~~~\gamma_5 S_{z} \otimes I_{z} + \gamma_6 \left(S_{x} \otimes I_{x} + S_{y} \otimes I_{y}\right)
\end{align}

where $S_{i}/I_{i}$ are spin-1 operators. Under the secular approximation, we can assume a large zero field splitting $\gamma_{1} \gg \gamma_6$ and thus ignore the last term in the above equation.  \cite{DOHERTY20131}. Also, it was shown in \cite{PhysRevLett.121.060401} that the effect of ambient environment of carbon atoms and other impurities in the diamond is effectively Markovian. The non-Markovian effects are thus attributed to a Hamiltonian that satisfies Theorem~\ref{Thm:sufficient_Hamiltonians}.

\noindent \section{Conclusion} 
We have characterized a class of Hamiltonians that generate classical memory processes represented as probabilistic mixtures of unitary Markovian processes. These Hamiltonians admit a complete set of commuting conserved quantities in the environment. Intuitively, fixing the conserved quantities induces an effective system Hamiltonian, which can be different for the different conserved values. As the environment is not observed, its initial state defines a classical probability distribution over the values of the conserved quantities, leading to a classical mixture of unitary system evolutions. 

Using the process matrix approach, we have demonstrated that several previously studied non-Markovian interactions, including non-trivial entangling interactions, are examples of these mixed unitary processes. 
We have also characterized the quantum circuits associated with these classes of processes. 
Particularly, we have shown that the multitime processes generated by these Hamiltonians can be described as sequences of controlled unitaries conditioned on a common environment in a fixed basis. 
Conversely, we have also shown that any mixed unitary process can be generated by a sequence of controlled unitaries and an initial quantum-classical system-environment state. 
This aligns with previous results of Ref.~\cite{Baecker2024} where controlled unitaries were linked to classical memory in the context of dynamical maps, which in our formalism corresponds to the constrained scenarios of only two-time correlations. 
Additionally, two-time mixed unitary processes, sometimes called random unitaries, are also known to characterize the class of recoverable channels whose noise evolution can be corrected perfectly upon measuring the environment and conditionally post-processing the system \cite{Gregoratti2003, Buscemi2005-ir, Buscemi2007}. Moreover, it has been shown that single-shot error correction schemes are compatible with stochastic noise that is spatially
local but permits arbitrary classical temporal correlations \cite{PhysRevX.6.041034,liu2024nonmarkoviannoisesuppressionsimplified}. Non-Markovian noise with classical memory can also be shown to suppress worst-case gate errors for certain classes of interaction Hamiltonians \cite{srivastava2025blindspotsrandomizedbenchmarkingtemporal}, while it has been demonstrated that operationally accessible measures of the non-classicality of a non-Markovian process are both computationally and experimentally feasible \cite{PhysRevX.10.041049}.
Furthermore, beyond mixed unitary processes, in Appendix~\ref{App:unitary_form_classical_memory}, we have provided quantum circuit models for the most general classical-memory processes.
 
 As future work, exploring other Hamiltonians that generate distinct types of non-Markovian memory and extending these results to continuous-time and continuous-variable processes are promising directions. Additionally, as shown in Appendix~\ref{App:Hamiltonian_entangled}, processes may exhibit classical memory at specific probing times despite having quantum memory otherwise. Investigating system-environment models that consistently exhibit similar memory characteristics across all probing times presents another compelling avenue for research.

\acknowledgments
{K.G. thanks Manabendranath Bera for the interesting discussions. This work is supported by the Hong Kong Research Grant Council (RGC) through grant No. 17307520, John Templeton Foundation through grant 62312, “The Quantum Information Structure of Spacetime” (qiss.fr) and the Australian Research Council (ARC) Centre of Excellence for Quantum Engineered Systems grant (CE170100009). A.K.R and V.S acknowledge funding from the Sydney Quantum Academy. C. G. was supported by a UTS Chancellor's Research Fellowship.}

\bibliography{Non_Markovianity}

\appendix

\section{Preliminaries} 
\label{Sec:Preliminaries} 
\subsection{Process matrix formalism}
We use the process matrix formalism~\cite{oreshkov12} to capture non-Markovian dynamics. The framework models a general scenario where a system and an environment, possibly initially correlated, can evolve and interact. The system can be probed at well-defined times during the evolution, a total of $N$ times. After the initial preparation and in between the probes, the system evolves according to some system-environment Hamiltonian that can be represented as a system-environment gate applied in the time step between two probing times. This leads to a general \emph{$N$-time process}, represented in Fig.~\ref{fig:gen_non-Markov} in the manuscript: initially correlated state, probe $1$, gate $1$, probe $2$, gate $2$, $\cdots$ gate $N-1$, probe $N$. 

Each probe $n$ is a quantum operation that is represented by a completely positive (CP) and trace non-increasing map $\mathcal{M}_{a_{n}|x_{n}}: A_{I}^{n}\rightarrow A_{O}^{n}$, that maps the input system $A_{I}^{n}$ to the output system\footnote{Note that, even in the typical scenario where operations act on a given system, input and output are modelled as distinct (isomorphic) spaces.} $A_{O}^{n}$. $A^n_{I(O)}\equiv \mathcal{L}(\mathcal{H}^{A^n_{I(O)}})$ is the space of linear operators on the Hilbert space $\mathcal{H}^{A^n_{I(O)}}$. In general, the operation involves a measurement, so $x_{n}$ and $a_{n}$ represent the classical settings and outcomes, respectively.  
Whenever there is no confusion, we will drop the labels for the Hilbert space. Furthermore, the framework assumes that the time to implement the operations is negligible compared to the time scales of the system-environment interaction.  We will use the term ``site'' to refer to the pair of input and output spaces, $A^n_I$, $A^n_O$, representing the possibility of an operation at time $n$. Sites can be seen as a temporal generalization of subsystems, see also \cite{Aharonov2014}.

Using the Choi Jamio{\l}kowski (CJ) isomorphism \cite{choi_completely_1975,jamio72,choi75b}, the operations can be equivalently represented by a linear, positive semi-definite operator $M_{a_{n}|x_{n}}$ defined on $A_I^{n}{\otimes} A_O^{n}$
\begin{align}
  M_{a_{n}|x_{n}} := [\mathcal{I}^{A_I^{n}}\otimes \mathcal{M}_{a_{n}|x_{n}}\left(\KetBra{\mathds{1}}{\mathds{1}}\right)^{A_I^{n}}], \label{Eq:CJ}
\end{align} 
where $\mathcal{I}^{A_I^{n}}$ is the identity map $A_I^{n}\to A_I^{n}$, $\Ket{\mathds{1}}^{A_I^{n}} := \sum_{i}\ket{i}^{A_{I}^{n}}\otimes \ket{i}^{A_{I}^{n}}$ is the unnormalised maximally entangled state, and $\{|i\rangle\}_i$ is an orthonormal basis for $A^{n}_{I}$. In general, the trace non-increasing property of the CP map is reflected by the operator inequality $\Tr_{A_O^{n}} M_{a_{n}|x_{n}}\le \mathds{1}^{A_I^{n}}$, note $\mathds{1}^{A_I^n}$ is the CJ representation of the trace-map acting on $A_I^n$. The equality is satisfied when the quantum operation is a completely positive, trace-preserving (CPTP) map, i.e., a quantum channel, representing a deterministic operation with no associated measurement output. A quantum instrument is a collection of CP maps summing up to a CPTP one, $\Tr_{A_O^{n}} \sum_{a_{n}} M_{a_{n}|x_{n}} =  \mathds{1}^{A_I^{n}}$, representing all the possible outcomes of a measurement. The last site consists only of an input space $A^N_I$, with no corresponding output space, and the associated instrument reduces to a Positive Operator-Valued Measure (POVM) acting on $A_I^{N}$, namely a collection of positive-semidefinite operators satisfying $\sum_{a_{N}} M_{a_{N}|x_{N}} =  \mathds{1}^{A_I^{N}}$.  

In the setup we are considering, the environment is never measured, and its effects only manifest as multi-time correlations between measurements on the system. All information accessible through operations on the system is compactly encoded in a \textit{process matrix}, Fig.~\ref{fig:gen_non-Markov}(a), which is a positive-semidefinite operator $W \in A_I^1 \otimes A_O^2\otimes \cdots A_I^N$, defined on the tensor product of all input-output Hilbert spaces. Together with the sequence of CP maps $\left\{\mathcal{M}^{A_I^{n}\to A_O^{n}}_{a_n|x_n}\right\}_{n=1}^N$ at different times, the process matrix $W^{A_I^1A_O^1\cdots A_I^N}$ generalises the Born rule~\cite{Shrapnel_2018} to evaluate the conditional probability of classical outcomes $\{a_n\}_n$ given settings $\{x_n\}_n$: 
\begin{align}
  &p(a_1, a_2,\cdots a_N|x_1, x_2,\cdots x_N)\nonumber \\
  &\qquad \quad =\Tr\left[\bigotimes_{n=1}^{N}\left(M^{A_I^{n}A_O^{n}}_{a_n|x_n}\right)^T W^{A^{1}_{I} A^{1}_{O}\cdots A_I^{N}}\right], \label{Eq:Born_rule1}
\end{align}
where the superscript $T$ represents transpose. The process matrix $W^{A^{1}_{I} A^{1}_{O}\cdots A_I^{N}}$ represents an $N$-time process and the operations occur at the $N$ sites. This formalism is equivalent to that of quantum channels with memory \cite{Kretschmann2005}, quantum strategies~\cite{gutoski06,gutoski2012quantum}, process tensors~\cite{Pollock_pra, pollock_operational_markov} and quantum combs~\cite{chiribella08, chiribella09b}.

\subsection{Markovian process}

We have briefly introduced this class of process in the main text, which we reintroduce with a bit more details. A Markovian process is characterized by memoryless dynamics, i.e., the system's future evolution does not depend on past states of the system. Operationally, a memoryless process should be reproducible without retaining any external memory. If the system interacts with an environment during its evolution, the environment can be discarded and re-prepared in a new state after each time-step, resulting in a sequence of quantum channels independent from each other. Fig.~\ref{fig:gen_non-Markov} shows a general non-Markovian (a) and a Markovian (b) process, where we have omitted the environment labels for clarity. Denoting $\rho$ the initial state and $\mathcal{T}^n:A_O^n\rightarrow A^{n+1}_I$ the quantum channels for $n=1,\dots, N-1$, the probability for a sequence of CP maps $\mathcal{M}_{a_1}, \dots, \mathcal{M}_{a_n}$ (where we omit the settings for simplicity) is simply given by the map composition $P(a_1,\dots,a_N ) = \Tr\left[\mathcal{M}_{a_N}\circ \mathcal{T}^{N-1} \cdots \circ \mathcal{T}^{1}  \circ \mathcal{M}_{a_1}(\rho)\right]$. This corresponds to plugging into the Born rule, Eq.~\eqref{Eq:Born_rule1}, the process matrix,
\begin{equation}\label{Markovian_process}
W_{\textrm{M}}  = \rho^{A_{I}^{1}}\otimes{{T}^{1}}  \otimes  \cdots  \otimes {{T}^{N-1}},
\end{equation}

where $T^{n} \in A_{O}^{n}\otimes A_{I}^{n+1}$ is the Choi representation of the quantum channel $\mathcal{T}^{n}$ and we omit the superscripts denoting the systems. We see that Markovian process matrices are proportional to product states, subject to the constraint that each factor is the Choi representation of a CPTP map.

A particular case that will be of interest later is when each channel is unitary, representing a closed-system evolution with no interaction with the environment. Accordingly, we define \emph{unitary Markovian process matrices} as\footnote{The term unitary process matrix has been used slightly differently in the literature, referring to processes that start with an output space \cite{Araujo2017purification, Costa2020}. Furthermore, here, we restrict unitary processes to be causally ordered and Markovian, which is not necessarily true in general.}
\begin{align}\label{Unitary_process}
& W_{\textrm{U}}  &  = \rho^{A_{I}^{1}}\otimes {\Proj{\widetilde{U}^{1}}} \otimes \cdots \otimes {\Proj{\widetilde{U}^{N-1}}},
\end{align}
where the CJ representation of the unitary $\widetilde{U}^{n}: \mathcal{H}^{A_{O}^{n}}\rightarrow \mathcal{H}^{A_{O}^{n+1}}$ is $\Proj{\widetilde{U}^{n}}$
where
\begin{equation}
    \label{unitarychoivector} 
    \Ket{\widetilde{U}^{n}} :=\id\otimes \widetilde{U}^{n} \Ket{\id} \in \mathcal{H}^{A_{O}^{n}}\otimes \mathcal{H}^{A_{O}^{n+1}}.
\end{equation}

\subsection{Classical-memory process}\label{Classical memory process}

We are interested in the subclass of non-Markovian quantum processes that can be simulated by a classical feed-forward mechanism through the environment, as depicted in Fig.~\ref{fig:gen_non-Markov}(c). This feed-forward mechanism can be modelled by a set of conditional quantum instruments in between the sites, where each instrument is conditioned on a setting that has a stochastic dependence on all the previous measurements and settings of past instruments \cite{Giarmatzi2021witnessingquantum}. In the end, all the classical outcomes that were used to implement the process are discarded. A process that can be reproduced in this way is said to have \emph{classical memory} because it only requires keeping track of classical variables that correlate the channels.

In particular, we consider an initial set of states $\rho_{s_{0}}$, each of which can be selected with a probability $p(s_{0})$. The classical variable $s_{0}$ also determines a probability distribution $p(s_{1}|s_{0})$, which in turn determines the instrument $T^n_{s_n}:=\{T^1_{m_{1}|s_{1}}\}_{m_1}$. All the classical variables at this stage stochastically determine the next instrument, and so on. The process matrix has the form:

\begin{align}\label{classical_mem_proc}
W_{\textrm{CM}} 
= &\sum_{\vec{s},\vec{m}} p(s_{0})  \rho_{s_{0}}  \otimes   p(s_{1}|s_{0}) T^1_{m_{1}|s_{1}} \otimes\cdots \nonumber\\
 &\quad\otimes  p(s_{N-1}|\vec{s}_{|N-2},\vec{m}_{|N-2}) T^{N-1}_{m_{N-1}|s_{N-1}}
\end{align}
where $\vec{s} \equiv \{s_{0},s_{1},\cdots s_{N-1}\}$ and $\vec{s}_{|n} \equiv \{s_{0},s_{1}\cdots s_{n}\}$ are the settings of the instruments, and $\vec{m} \equiv \{m_{1},m_{2}\cdots m_{N-1}\}$ and $\vec{m}_{|n} \equiv \{m_{1},m_{2}...m_{n}\}$ are the measurement outcomes of the instruments. This form of process matrix then represents the most general type of multi-time quantum process with classical memory.

\subsection{Classical common-cause process}\label{CCC sec}
We now reintroduce an interesting sub-class of classical-memory processes, called the \emph{classical common-cause processes (CCC)}. Notice that in Eq.~\eqref{classical_mem_proc}, if the conditional probability distribution is independent of measurement results $\vec{m}$ of instruments, i.e., $p(s_{n}|\vec{s}_{|n-1},\vec{m}_{|n-1})=p(s_{n}|\vec{s}_{|n-1})$ for all $n$, then the process matrix takes the form:

\begin{align}\label{CCC_proc1}
&W_{\textrm{CCC}} \nonumber
\\
&=\sum_{\vec{s}} p(s_{0},s_{1} \cdots , s_{N-1}) \rho_{s_{0}} \otimes T^1_{s_{1}}\otimes  \cdots \otimes T^{N-1}_{s_{N-1}},
\end{align}
where $T^{n}_{s_{n}}:=\sum_{m_n}T^{n}_{m_{n}|s_{n}}$ is a CJ representation of a CPTP map. 
From now on for notational convenience we represent $W_{\textrm{CCC}}$ as 
\begin{align}\label{CCC_proc2}
W_{\textrm{CCC}} &= \sum_{\nu} p(\nu)\rho_{\nu} \otimes T^{1}_{\nu} \otimes  \cdots \otimes T^{N-1}_{\nu} , \\
&=\sum_\nu p(\nu) \widetilde{W}_\nu
\end{align}
where in general, $\nu\equiv \{s_1,s_2,\cdots, s_{N-1}\}$ represents multiple random variables. Note, $W_{\textrm{CCC}}$ is a convex sum of Markovian process matrices $\widetilde{W}_\nu = \rho_\nu \bigotimes_{n=1}^{N-1} T^n_\nu$.

A particular case of CCC processes relevant to our results is that of \emph{mixed unitary processes}, whose process matrices are probabilistic mixtures of unitary process matrices, as defined in Eq.~\eqref{Unitary_process}:

\begin{equation}
W_{\textrm{MU}} \nonumber
= \sum_{\nu} p(\nu)\rho_{\nu} \otimes \Proj{\widetilde{U}^{1}_{\nu}} \otimes \cdots \otimes \Proj{\widetilde{U}^{N-1}_{\nu}} \label{Eq:Mixed-Unitary process}.
\end{equation}
For $N{=}2$, such a process has been experimentally realised~\cite{Goswami_non_Markov}, where the source of classical memory was due to the non-trivial joint probability distribution of two random variables.

It turns out that, for $N{=}2$, all classical-memory processes are CCC~\cite{Giarmatzi2021witnessingquantum}, however for $N{\ge} 3$, CCC processes are a strict subset of those with classical memory~\cite{taranto2023characterising}.  We see from Eq.~\eqref{CCC_proc1} that CCC processes are proportional to separable states. Furthermore, in Ref.~\cite{Nery_2021}, examples of a strictly larger class of separable processes have been found by removing the trace-preserving condition from $T^n_{s_n}$ in Eq.~\eqref{CCC_proc1}. It is also possible to have quantum memory processes~\cite{Giarmatzi2021witnessingquantum,Nery_2021,taranto2023characterising}, the temporal correlation manifested in such a process exhibits genuine nonclassical correlation. In the present work, we do not focus on these processes.

\section{Proof of Theorem~\ref{Thm:sufficient_Hamiltonians}}\label{App:proof_theorem_suff_hamil}

\begin{proof}
Consider the system of interest in the process is accessed at time steps $\{t_n\}_{n=1}^N$.  The system-environment unitary $U^n$ from $t_n$ to $t_{n+1}$ is given by a time-ordered exponential as shown in Eq.~\eqref{Eq:nunitary}. Given that the environment operators commute, we can always find a common eigenbasis \{$\ket{\nu}$\} so that $\mathcal{E}_{j}(t){=}\sum_\nu \lambda^{\nu}_j(t)\ketbra{\nu}{\nu}$ and $\lambda^{\nu}_j(t)$ are the time-dependent eigenvalues. The corresponding unitary $U^{n}$ takes the form of a controlled unitary:

\begin{align}
U^n &= \sum_{\nu}T \exp \left[-i \int_{t_n}^{t_{n+1}}dt \tilde{S}^{\nu}(t)\right] \otimes \proj{\nu} \nonumber  \\
&=\sum_\nu \widetilde{U}_{\nu}^{n}{\otimes}\ketbra{\nu}{\nu} \label{Eq:unitary_hamiltonian}
\end{align}
where $\tilde{S}^{\nu}(t) = \sum_j \lambda^{\nu}_j(t) S_j(t)$ and
\begin{align}
 \widetilde{U}_{\nu}^{n}&=T \exp \left[-i \int_{t_n}^{t_{n+1}}dt \tilde{S}^{\nu}(t)\right]\label{Eq:choi_nth_unitary}. 
\end{align}
The time-dependence of the operators $U^n$ and $\widetilde{U}^n$ is implied. The CJ-representation of the unitary in Eq.~\eqref{Eq:unitary_hamiltonian}, labelled by the relevant input-output Hilbert spaces, is
\begin{align}
&\Proj{U^{n}}^{A_O^{n}A_I^{n+1}E^{n}E^{n+1}} \nonumber \\
  &{=}\sum_{\nu, \widetilde{\nu}}\KetBra{\widetilde{U}_\nu^{n}}{\widetilde{U}_{\widetilde{\nu}}^{n}}^{A_O^{n}A_I^{n+1}} {\otimes}\ketbra{\nu \nu}{\widetilde{\nu} \widetilde{\nu}}^{E^{n}E^{n+1}} \label{Eq:choi_unitary_nth_step}.
\end{align}

Let us order the tensor products in reverse for the convenience of notation and consider the CJ-representation of the unitary $\Proj{U^{N-1}}$ at the \emph{final} time step $N$, whose output environment, $E^{N}$, is traced out resulting in
\begin{align}
&\mathds{1}^{E^{N}} *\Proj{U^{N-1}}^{A_O^{N-1}A_I^{N}E^{N-1}E^{N}} \nonumber \\
&= \sum_{\nu} \Proj{\widetilde{U}_{\nu}^{N-1}}^{A_O^{N-1}A_I^{N}}{\otimes} \ketbra{\nu}{\nu}^{E^{N-1}},  \label{Eq:trace_last_unitary} \end{align}
where $\Proj{\widetilde{U}_{\nu}^{n}}{=}\KetBra{\widetilde{U}_{\nu}^{n}}{\widetilde{U}_{\nu}^{n}}$ is the CJ representation of the unitary defined in Eq.~\eqref{Eq:choi_nth_unitary}.

Next, we consider the term $\Proj{U^{N-2}}$. We expand it using Eq.~\eqref{Eq:choi_unitary_nth_step}, and then calculate the link product with Eq.~\eqref{Eq:trace_last_unitary} to obtain
\begin{align}
    &\mathds{1}^{E^{N}} *\Proj{U^{N-1}} * \Proj{U^{N-2}}  \nonumber \\
    =\sum_\nu &\Proj{\widetilde{U}^{N-1}_\nu}^{A_O^{N-1}A_I^{N}} {\otimes} \Proj{\widetilde{U}^{N-2}_\nu}^{A_O^{N-2}A_I^{N-1}} \nonumber \\&{\otimes}\ketbra{\nu}{\nu}^{E^{N-2}}\label{Eq:last_two_unitary}.
\end{align}
Continuing the concatenation till the first timestep and finally appending the initial system-environment state $\rho^{A_{I}^{1}E^{1}}$, we have the overall process matrix
\begin{align}
&W^{A^{1}_IA^{1}_O{\dots}A^{N}_I} \nonumber \\
&=    \bigg(\Asterisk_{n=1}^{N-1} \Proj{U^{n}}^{A_O^{n}A_I^{n+1}E^{n}E^{n+1}}\bigg) * \mathds{1}^{E^{N}}* \rho^{A_{I}^{1}E^{1}} \nonumber \\ 
&=\sum_\nu p(\nu){\bigotimes_{n=1}^{N-1}}\left(\Proj{\widetilde{U}^{n}_\nu}^{A_O^{n}A_I^{n+1}}\right) {\otimes} \rho_{\nu}^{A_{I}^{1}}\nonumber \\
&= \sum_{\nu} p(\nu) \widetilde{W}_\nu^{A^{1}_IA^{1}_O{\dots}A^{N}_I},
\end{align}
where the conditional state $\rho_{\nu}^{A_{I}^{1}}$ is given by
\begin{equation}
  \rho_{\nu}^{A_{I}^{1}} =\frac{ \Tr_{E^{1}}[\rho^{A_{I}^{1}E^{1}}(\mathds{1}^{A_{I}^{1}}{\otimes}\ketbra{\nu}{\nu}^{E^{1}})]}{p(\nu)},   
\end{equation}
and where
\begin{equation}
  p(\nu) = \Tr_{A_I^{1}E^{1}}[\rho^{A_{I}^{1}E^{1}}(\mathds{1}^{A_{I}^{1}}{\otimes}\ketbra{\nu}{\nu}^{E^{1}})]
\end{equation}
forms a probability distribution and 
\begin{equation}
 \widetilde{W}_\nu{=} \rho_{\nu}^{A_{I}^{1}}\bigotimes_{n=1}^{N-1}\Proj{\widetilde{U}_{\nu}^{n}}   
\end{equation}
 is a unitary Markovian process. This concludes our proof.
\end{proof}

\section{Proof of Theorem~\ref{Thm:direct_cause_extension}}\label{App:proof_direct_cause_ext}

\begin{proof}  Consider an arbitrary $N$ time-steps mixed unitary process $W_{\textrm{MU}}$, with the form in Eq.~\eqref{Eq:Mixed-Unitary process}, where $\Proj{\widetilde{U}_{\nu}^{n}}$ is the CJ representation of the unitary $\widetilde{U}_{\nu}^{n}:A_O^n \to A_I^{n+1}$ for $n \in \{1,2,{\cdots},N{-}1\}$. 
We can now introduce the controlled unitaries
\begin{equation}
    \widetilde{U}^{n} = \sum_{\nu} \widetilde{U}_{\nu}^{n} \otimes \proj{\nu}^{E},
\end{equation}
where $\{\ket{\nu}\}$ is an orthonormal basis in $E$. 
Consider now a set $\Gamma$ of the initial system-environment states, such that 
\begin{align}
    \Gamma {:=} \{\rho^{A^1_IE}|\Tr_E[({\mathds{1}_{A^1_I}{\otimes} \proj{\nu}^E})\rho^{A^1_IE}]{=}p(\nu)\rho_\nu, {\forall} \nu\}. \label{Eq:steered states}
\end{align}
 An example of such a state is the following quantum-classical state:

\begin{align}
    \rho^{A_{I}^{1}E} = \sum_\nu p(\nu) \rho_\nu^{A_{I}^{1}} \otimes\proj{\nu}^E. 
\end{align}
Now, with any state $\rho\in \Gamma$, the mixed unitary process $W_{\textrm{MU}}$ assumes the form

\begin{align}
    W_{\textrm{MU}}= &\rho *\Proj{\widetilde{U}^{1}} *\cdot\cdot\cdot *\Proj{\widetilde{U}^{N-1}} * \mathds{1}^{E}. 
\end{align}
 This proves our theorem. 
\end{proof}

\begin{figure}
    \centering
    \includegraphics[width=0.9\columnwidth]{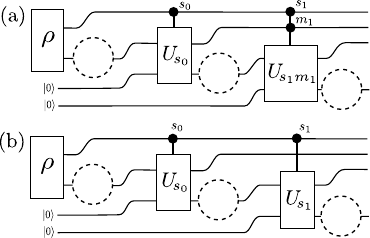}
    \caption{Circuits representing different classes of classical memory, with time oriented from left to right. (a) Full classical memory with dependence on past (including stochastic), and (b) classical common cause memory. The dashed circles represent the sites.}
    \label{fig:memory_circuits}
\end{figure}

\section{Unitary form of classical memory non-Markovian processes} \label{App:unitary_form_classical_memory}

Here, we show that every multi-time classical memory non-Markovian process can be written as a sequence of controlled unitaries and \textit{vice versa} --- any sequence of controlled unitaries where the controls are in a common basis constitutes classical memory. In Section~\ref{Classical memory process}, we gave the general mathematical form of a classical-memory process as Eq.~\eqref{classical_mem_proc}, which represents arbitrary instruments which act probabilistically conditioned on past classical variables. This is the basis on which we claim that memory is classical in nature.
This form can be simplified by combining terms to get
\begin{equation}\label{WCMappend}
W_{\textrm{CM}} = \sum_{\vec{s},\vec{m}} p(s_{0}) \rho_{s_{0}} \bigotimes_{n=1}^{N-1}   \bar{T}^{n}_{m_{n}s_{n}|\vec{s}_{|n-1},\vec{m}_{|n-1}},
\end{equation}
where
\begin{equation}
 \bar{T}^n_{m_{n}s_{n}|\vec{s}_{|n-1},\vec{m}_{|n-1}} = p(s_{n}|\vec{s}_{|n-1}\vec{m}_{|n-1}) T^{n}_{m_{n}|s_{n}}
\end{equation}
are also instruments with measurement `outcomes' labelled by both $m_{n}$ and $s_{n}$, since $\sum_{s_{n}, m_{n}}\tilde{T}^n_{m_{n}s_{n}|\vec{s}_{|n-1},\vec{m}_{|n-1}}$ are CPTP channels. This shows that the non-deterministic instruments can be modelled by suitable deterministically applied instruments.

For completeness, we explicitly construct a set of controlled unitaries that would yield the same $W_{\textrm{CM}}$ as \eqref{classical_mem_proc} given the set of conditional probabilities $p(s_{n}|\vec{s}_{|n-1},\vec{m}_{|n-1})$, and corresponding instruments.
First note that a quantum instrument can be specified by a set of Kraus operators $\{M_{m,k}\}$ where a measurement result of $m$ applies the CP channel $\sum_{k} M_{mk}\cdot M_{mk}^{\dagger}$, and $\sum_{mk}M_{mk}^{\dagger}M_{mk}=\id$.
The index $k$ models inefficient measurements.
We can construct a (non-unique) isometry $V=\sum_{mk}\ket{k}\ket{m}M_{mk}$, and this isometry can then be extended to a unitary operator by supplying orthogonal columns.
The resulting unitary represents a possible dilation of the instrument to a system-environment model, with a projective measurement $\ketbra{m}{m}$ on the second environment.
Alternatively, combining the environments, we can label their joint initial state as $\ket{0}$, and recognize the measurement as a degenerate measurement with measurement operators $\id\otimes\ketbra{m}{m}$. So, every instrument can be dilated to a conditional unitary with an extra environment and a measurement on that environment.

\begin{figure}
    \centering
    \includegraphics[width=0.95\columnwidth]{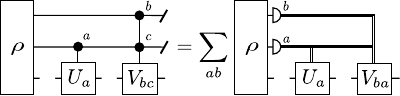}
    \caption{Tracing over controlled unitaries in a fixed basis is equivalent to a sum over classically controlled unitaries. Time is from left to right. Note that controls on the same environment become conditioned on the same classical index.}
    \label{fig:control_to_classical}
\end{figure}

Each of these classically-conditioned unitaries can then be modelled as a controlled unitary, controlled by some environment with basis $\{\ket{s_{k}}\}$ as $\sum_{s_{k}} \ketbra{s_{k}}{s_{k}}\otimes U_{s_{k}}$. Provided all the controls are in the same basis (and respecting any degeneracies in the measurements), tracing over the environment will result in a sum over classically-conditioned unitaries as in Fig.~\ref{fig:control_to_classical}.

To explicitly model the \emph{stochastic} dependence on previous variables consider the following unitary operator

\begin{equation}
  R_{a}=\ket{\xi}\bra{0}+\sum_{i=1}^{d-1} \ket{\xi_{\perp}^{i}}\bra{i}
\end{equation}
where $\ket{\xi}=\sum_{i=0}^{d-1} \sqrt{p(r_{i}|a)} \ket{s_{i}}$, $\braket{j}{k}=\delta_{jk}$, and $\{\ket{s_{i}}\}_{i=0}^{d-1}$ is the preferred environment orthonormal basis.
The $\ket{\xi_{\perp}^{i}}$ are chosen so that  $\braket{\xi}{\xi_{\perp}^{i}}=0$, and $\braket{\xi_{\perp}^{i}}{\xi_{\perp}^{j}}=\delta_{ij}$.
Now, the controlled-unitary circuit in Fig.~\ref{fig:stochastic_to_control} effectively implements a stochastically conditioned unitary as required upon tracing out the environment.

Hence any process matrix of the form of Eq.~\eqref{classical_mem_proc} can be modelled by a quantum circuit of the form depicted in Fig.~\ref{fig:memory_circuits}~(a). The reverse is easy to show --- any circuit of the form in Fig.~\ref{fig:memory_circuits}~(a) will have a process matric of the form of Eq.~\eqref{classical_mem_proc}.

\begin{figure}
    \centering
    \includegraphics[width=0.8\columnwidth]{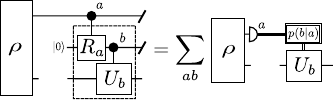}
    \caption{Implementing a stochastically controlled unitary. The operator in the dashed box is also a unitary. Time is from left to right.}
    \label{fig:stochastic_to_control}
\end{figure}

\emph{Classical common cause:} As discussed in Section~\ref{CCC sec}, if each of the conditional probability distributions is independent of the measurement outcomes, 
$p(s_{n}|\vec{m}_{|n-1}\vec{s}_{|n-1})=p(s_{n}|\vec{s}_{|n-1})$, 
we have a CCC process given in Eq.~\eqref{CCC_proc1}. The corresponding circuit diagram for this class becomes Fig.~\ref{fig:memory_circuits}~(b).

\section{Hamiltonian with entangled eigenstates} \label{App:Hamiltonian_entangled}

In this section, we consider an interaction Hamiltonian with non-product eigenstates and analyze the nature of non-Markovianity, i.e., whether the resulting process is classical or quantum memory process. To be concrete, we consider the isotropic Heisenberg model with magnetic field in $z$ direction,
\begin{eqnarray}\label{Heisenberg_int}
     H^{ES} = J(\sigma_{x}\otimes\sigma_{x} + &\sigma_{y}\otimes\sigma_{y}+\sigma_{z}\otimes\sigma_{z})\\&+ B(\sigma_{z}\otimes\mathds{1}+\mathds{1}\otimes\sigma_{z})\nonumber
\end{eqnarray}
The eigenvectors and corresponding eigenvalues for the above interaction are the following
\begin{eqnarray}
    |\psi_{1}\rangle = \frac{1}{\sqrt{2}}(|01\rangle-|10\rangle);\quad E_{1} = -3J\\
    |\psi_{2}\rangle = |11\rangle;\quad  E_{2} = J-2B\\
    |\psi_{3}\rangle = \frac{1}{\sqrt{2}}(|01\rangle+|10\rangle);\quad E_{3} = J\\
    |\psi_{4}\rangle = |00\rangle;\quad E_{4} = J+2B.
\end{eqnarray}
Evidently, two of the eigenstates are maximally entangled states, and for a non-zero magnetic field ($B\ne 0$), the spectrum is non-degenerate.

Now, the unitary generated by the interaction $H$ is
\begin{eqnarray}\label{Heisenberg_unitary}
    U(t) = \sum_{i=1}^{4}\exp{(-iE_{i}t)}|\psi_{i}\rangle\langle\psi_{i}|,
\end{eqnarray}
where $t$ is the interaction time. Importantly, even though the spectrum of interaction Hamiltonian is non-degenerate, there can be degeneracy in the eigenvalues of the unitary (the phase terms $\exp(-iE_{i}t)$). In particular, for the following interaction times
\begin{eqnarray}\label{int_times}
    t = \frac{n\pi}{2J}; \quad n\in \mathds{N}
\end{eqnarray}
the phase terms corresponding to $|\psi_{1}\rangle$ and $|\psi_{3}\rangle$ are equal, hence any quantum state in the subspace spanned by $\ket{\psi_1}$ and $\ket{\psi_3}$ is a valid eigenstate. Particularly, choosing eigenstates $1/\sqrt{2}(\ket{\psi_3}\pm \ket{\psi_1})$, we observe that at these specific interaction times, the same unitary can be generated using an alternative interaction Hamiltonian with the same spectrum but with all product eigenstates. This allows us to express the alternative Hamiltonian as a sum of tensor products. Due to the orthogonality of the product eigenstates, the operators associated with the environment naturally commute. Hence, we can apply Theorem \ref{Thm:sufficient_Hamiltonians}. To conclude, the interaction Hamiltonian Eq.~\eqref{Heisenberg_int} results in a classical common-cause process if the interaction times between the interventions are as in Eq.~\eqref{int_times}. However, the resulting process becomes a quantum memory process for a different interaction time. To show that, we consider the simplest scenario with only two-time interventions. Consider the initial system-environment state to be a correlated Bell state, i.e., $\rho^{E_{1}A_{1}} = |\psi^{+}\rangle\langle\psi^{+}|$, where $|\psi^{+}\rangle = (|00\rangle+|11\rangle)/ \sqrt{2}$. The resulting process is obtained as
\begin{eqnarray}
    W^{A_{1}A_{2}B_{1}} = \rho^{E_{1}A_{1}}*\Proj{U(t)}^{E_{1}A_{2}E_{2}B_{1}}*\mathds{1}^{E_{2}},
\end{eqnarray}
where $\Proj{U(t)}^{E_{1}A_{2}E_{2}B_{1}}$ is the CJ operator corresponding to the unitary in Eq. \eqref{Heisenberg_unitary}. We use the Peres-Horodecki criterion~\cite{peres96,horodecki96} in the bipartition $A1|A_2B_1$ of the normalized process $W^{A_{1}A_{2}B_{1}}$ to check whether the process is a quantum memory process. To this end, we use the `negativity' as the measure of quantum memory:
\begin{eqnarray}
    \mathcal{N} = \sum_{\lambda<0}|\lambda|,
\end{eqnarray}
where $\lambda$ are the eigenvalues of $W^{T_{A_{1}}}$, where, superscript $T_{A_{1}}$ is partial transpose with respect to $A_{1}$.
\begin{figure}
    \centering
    \includegraphics[width=1\linewidth]{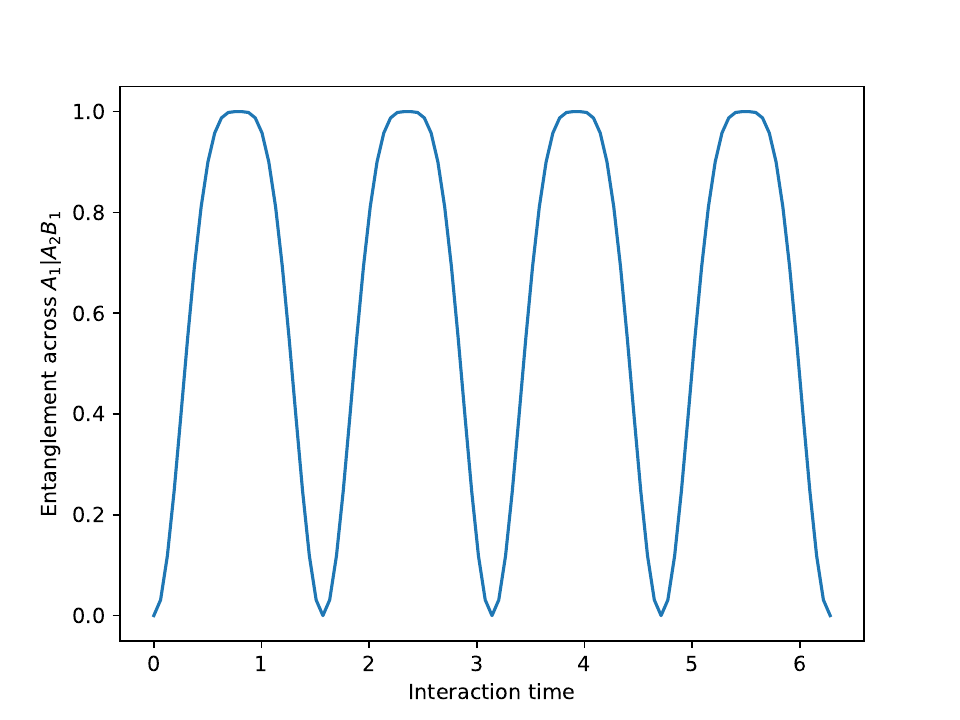}
    \caption{Entanglement across the bipartition $A_{1}|A_{2}B_1$ in the process matrix $W^{A_{1}A_{2}B_{1}}$, for varying interaction time in units of $Jt$.  \label{fig:QM_Heisenberg_int}} 
\end{figure}
In Fig. \ref{fig:QM_Heisenberg_int}, we have plotted the entanglement measure across $A_1|A_{2}B_1$ with varying interaction time. We observe the oscillating behaviour of entanglement, with $J$ determining the frequency of oscillation, and it vanishes for the interaction times $t = \pi/2,\pi,3\pi/2,2\pi$, where, as expected, it can be modelled as a classical common-cause process as discussed above.

This example shows that a Hamiltonian with entangled eigenstates can generate classical-memory processes but only for certain choices of interaction times that result in unitaries with product eigenstates. Otherwise, the process is a quantum memory process. This strengthens the conjecture that if an interaction generates a classical-memory process for all times, it must have only product eigenstates. Similar to the above example, one can construct other models, for instance, interaction with all Bell states as eigenstates, i.e, $H = E_{\pm}|\psi^{\pm}\rangle\langle\psi^{\pm}| + F_{\pm}|\phi^{\pm}\rangle\langle\phi^{\pm}|$. Then, a sufficient condition for the classical-memory process is that the ratio $E_{+}-E_{-}/F_{+}-F_{-}$ must be a rational number, i.e., with this constraint satisfied, there exist interaction times such that resulting process is a classical common-cause process. The Observation~\ref{Obs:entangled_eigenstate} in the main text is a special case of this scenario.

\section{Slowly varying environment basis} \label{App:varying_env-basis}
Here, we consider a two dimensional system and environment with three interventions on the system. We consider the control to be in the canonical basis $\{|0\rangle,|1\rangle\}$ between the first two intervention (i.e., $A$ and $B$). Then, we introduce a variation in the control basis in the environment in the next time step (between $B$ and $C$) by rotating the basis by $U(\delta t)$, which is a unitary generated by a slowly varying Hamiltonian. In the first time-step, we consider the controlled-NOT as the environment-system unitary, with the environment as control. Moreover, the Hamiltonian which generates rotation in environment basis is chosen to be Pauli $X$. For the resulting process, we observe the amount of entanglement across the lab $A$ and lab $B$ for different choices of $\delta t$. We find that for any non-zero $\delta t$ the resulting process has entanglement across lab $A$, hence requiring quantum memory across the lab (see Figure  \ref{fig:Mem_Var}).  
\begin{figure}[h!]
    \centering
    \includegraphics[width=1\linewidth]{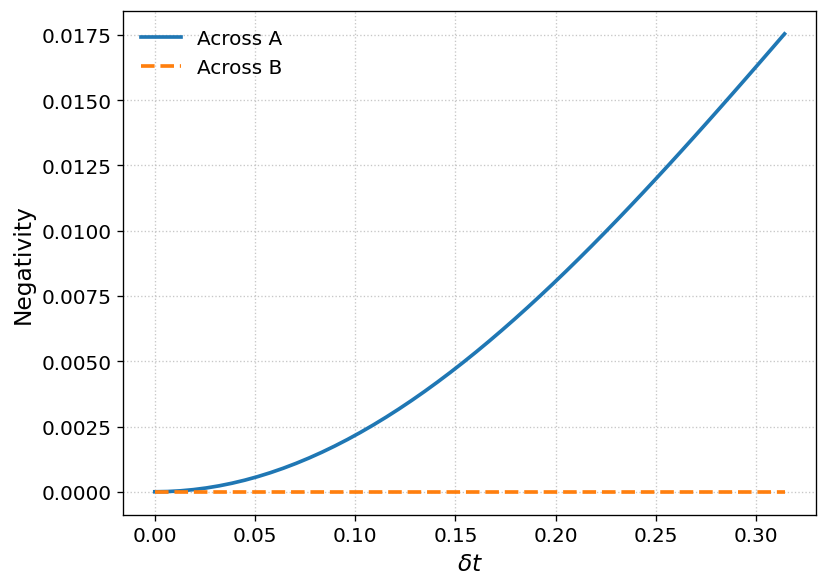}
    \caption{Quantum memory across lab $A$ and lab $B$ with variation in the control basis of the environment.}
    \label{fig:Mem_Var}
\end{figure}

\end{document}